\documentclass[twocolumn]{autart}
\usepackage{graphicx}
\usepackage{amsmath,amssymb,amsfonts}
\usepackage{subcaption}
\usepackage{adjustbox}
\graphicspath{{fig/eps/}{fig/pdf/}{fig/png}}
\usepackage{breqn}
\usepackage{xpatch} 
\makeatletter 
\xpatchcmd{\runningauthor@fmt}{\global\edef}{\protected@xdef}{}{}
\xpatchcmd{\runningauthor@fmt}{\global\edef}{\protected@xdef}{}{}
\xpatchcmd{\author@fmt}{\edef}{\protected@edef}{}{}
\def\@xnamedef#1{\expandafter\protected@xdef\csname #1\endcsname}
\def\ead@au#1{\protected@edef\@ead@au{#1}}
\def\add@xtok#1#2{\begingroup
  \protected@xdef\@act{\global\noexpand#1{\the#1#2}}\@act
\endgroup}
\def\no@harm{}
\makeatother

\begin{document}

\begin{frontmatter}

\title{Impedance Control via Generalized Output Regulation\thanksref{footnoteinfo}}

\thanks[footnoteinfo]{The material in this paper was not presented at any conference.}

\author[USP]{Hélio Jacinto Cruz Neto}\ead{helio.neto@usp.br},    
\author[USP]{Victor Shime}\ead{victor.shime@usp.br},               
\author[USP]{Thiago Boaventura}\ead{tboaventura@usp.br}  

\address[USP]{Department of Mechanical Engineering, São Carlos School of Engineering, University of São Paulo, Av. Trabalhador São Carlense, São Carlos, 13566-590, SP, Brazil}


\begin{abstract}                          
Achieving a desired impedance during physical interaction remains a central problem in compliant control. In many practical implementations, admittance control combined with linear position controllers is employed, but this structure typically results in only approximate impedance behavior. This paper establishes a rigorous equivalence between the impedance realization problem and the generalized output regulation theory, and shows that, under the admittance architecture and assuming the inner position controller is restricted to linear combinations of commonly used signals, a unique control law exists that achieves exact compliant control. The resulting controller introduces additional terms with respect to a conventional PD formulation, while preserving a simple structure. Numerical simulations are used to illustrate the effectiveness of the proposed approach under nominal conditions and in the presence of model uncertainties, and to compare its performance against a standard PD controller. The results demonstrate improved interaction performance and robustness, supporting the theoretical analysis.
\end{abstract}

\end{frontmatter}

\section{Introduction}
In robotics, interaction problems arise whenever a robot exchanges forces with the environment, objects, or humans, rather than simply following a predefined motion. Such situations occur in a wide range of domains, including industrial assembly \cite{park2017}, physical rehabilitation via exoskeletons \cite{Escalante2023,Li2018}, human-robot collaborative manipulation \cite{Jiao2022,Ta2021}, and legged locomotion on unstructured terrain \cite{Bermudez2025,Pedro2024}. As interaction appears across numerous robotic applications, properly handling it is therefore fundamental to ensure stability, robustness and safety in uncertain environments. 

A prominent  class of solutions for addressing interaction problems is based on compliant control. In this approach, the controller is designed to impose a specified  relation between measured interaction forces and the robot's motion, so that the system's response to external forces is governed by a chosen stable operator \cite{schumacher2019}, often a mass-spring-damper one. Embedding the desired compliant behavior directly into the closed-loop dynamics enables the controller to manage both free and constrained motion without the need for mode switching \cite{kang2009} or dealing with the problems of identifying constrained directions \cite{hogan2022}.  

Compliant control may be realized using either impedance or admittance formulations. Both approaches aim to achieve the same force–motion relationship; however, they implement it differently: impedance control maps motion deviations to commanded forces, whereas admittance control maps measured forces to desired trajectories \cite{Calanca2016}. In his seminal work \cite{Hogan1985}, Hogan suggested a feedback strategy to implement compliant control that belongs to the class of methods now referred to as impedance control. The resulting control law produced the desired impedance behavior under ideal conditions, but it relied on the complete robot model and offered no tunable gains to improve performance. For this reason, this approach has also been referred to in the literature as dynamics-based impedance control \cite{kang2009,Valency2003}. 

One alternative to the dependence on an exact model is the admittance formulation. The term admittance control dates back to \cite{newman1992}, although the same concept had already been implemented earlier under the name position-based impedance control \cite{lawrence1987}. By converting the impedance-tracking problem into a position-tracking one, admittance implementations can ensure robustness 
via a properly designed position controller \cite{kang2009}. Robust admittance implementations have been achieved using several control strategies, most commonly variants of sliding-mode and adaptive control  \cite{Dewit1997,MUJICA2023,Shen2020}. Though they improve robustness, these methods typically rely on an online model of the robot dynamics or on complex algorithms, which hinders their practical deployment.

Other common admittance implementations can be described by the general structure illustrated in Fig.~\ref{im:adm}, which is adapted from the review papers on compliant control \cite{Calanca2016,schumacher2019} and captures their shared essential properties. Although the position controller could use additional inputs, several works \cite{Calanca2017,Lawrence1988,Ott2010,Ott2015,Hauwermeiren2025,Volpe1993} model it as a rational transfer function (usually a variant of the PID controller). As depicted in Fig.~\ref{im:adm}, this representation implies that the position controller uses only the position error signal for feedback, with no additional inputs. While such implementations might render a passive impedance at the interaction port \cite{Calanca2017}, their ability to achieve the prescribed target impedance is not generally established. In fact, we show that the gap is not merely a matter of tuning: if the controller structure excludes additional signals (notably interaction force and acceleration-related feedback), then gain adjustment alone cannot achieve exact tracking of the prescribed impedance outside a limited achievable set. Moreover, whether a controller can passively render a given target stiffness relative to the environment’s stiffness typically depends on the inner-loop controller. For example, \cite{Calanca2017} shows that a modified impedance architecture incorporating acceleration feedback can, in theory, passively realize arbitrary target impedances. For other commonly used controllers, passivity requires the target stiffness to be no greater than the environment stiffness \cite{Vallery2008}.  

\begin{figure}[h!]
  \centering
  \includegraphics[scale = 0.60]{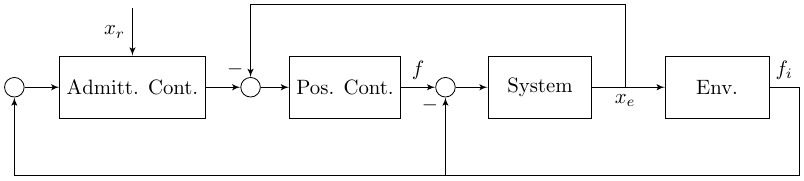}
\caption{Block diagram illustrating  the admittance structure.}\label{im:adm}
\end{figure}

Given this context, the objective of this paper is to determine which position controllers within the admittance architecture -- whose control laws are restricted to linear combinations of standard signals and their derivatives (joint position, interaction force, and admittance reference position) -- can realize the prescribed target impedance, and to compare the performance of those controllers with other commonly used position controllers. To this end, we employ the generalized output regulation theory \cite{saberi2001}, which extends classical output regulation to exosystems with inputs by incorporating disturbance decoupling results. Because the theory provides necessary and sufficient conditions for output regulation, casting the admittance control problem in this framework enables us to derive the \emph{unique} control law structure that attains the prescribed target impedance,  showing that commonly added force/acceleration feedback terms are necessary rather than heuristics.

The main contributions of this paper are: (i) derive rigorous conditions under which the impedance control problem can be formulated within the generalized output regulation framework; (ii) prove that, under the admittance architecture and assuming the inner position controller is restricted to linear combinations of commonly used signals and their derivatives (joint position, interaction force, and admittance reference position), there exists a unique linear control law structure that realizes the prescribed target impedance, which directly implies the necessity of these additional signals and establishes a fundamental tuning limit for incomplete controller structures; (iii) show how the derived control law can be extended to systems with soft joints, avoiding the issues associated with non-collocated control; (iv) compare the performance of the derived control law with a conventional PD controller.

\section{Generalized output regulation}

This section reviews the results of the generalized output regulation problem that are relevant for formulating the compliant control problem as an instance of output regulation. For further details, the reader is referred to \cite{saberi2001}.

Consider a linear system described by

\begin{subequations}\label{eq:wsystem}
\begin{align}
\dot{x} & = Ax+Bu+E w, \label{eq:system}\\
\dot{w} & = Sw+D r, \label{eq:exosystem}\\
e & = Cx+F w, \label{eq:error}
\end{align}
\end{subequations}
where \eqref{eq:system} describes the plant dynamics, \eqref{eq:exosystem} represents the exosystem, which may generate reference signals, disturbances, or both, and \eqref{eq:error} defines variables that must be controlled. Also, $x \in \mathbb{R}^n$ is the state, $u \in \mathbb{R}^m$ is the control input, $w \in \mathbb{R}^p$ is the exosystem state, $r \in \mathbb{R}^q$ is the exosystem input, and $e \in \mathbb{R}^s$ is the error. Capital letters denote matrices of real numbers with appropriate dimensions. 

The generalized output regulation problem consists in determining, if possible, a control input $u$ with a prescribed structure (e.g., output feedback, state feedback) such that the resulting closed-loop system is asymptotically stable and the error converges to zero as time tends to infinity, for any initial state and piecewise-continuous signals $r(t)$. The difference between the generalized and the standard \cite{huang2004} output regulation problems is that, in the latter, the exosystem has no inputs. In this paper, we review the relevant results for the case of state feedback controllers,

\begin{equation}\label{eq:sf}
    u = Kx+Lw.
\end{equation}

To properly frame the results, we introduce the following definition and assumptions.

\begin{defn}
    The stabilizable weakly unobservable subspace $\mathcal{V}(A,B,C)$ is the maximal subspace of $\mathbb{R}^n$ that is $\left(A+BK\right)$-invariant and contained in $\mathrm{ker}\left(C\right)$ such that the eigenvalues of $\left(A+BK\right)|\mathcal{V}$ have negative real parts for some $K$.
\end{defn}
\begin{assum}\label{ass:stab}
The pair $(A,B)$ is stabilizable.
\end{assum}
\begin{assum}\label{ass:antiHur}
$S$ has no eigenvalues with negative real parts.
\end{assum}

Given these conditions, the theorem stated below is an immediate consequence of Theorem 4 in \cite{saberi2001}:

\begin{thm}\label{theor:exactreg}
Consider that the system in \eqref{eq:wsystem} satisfies Assumptions \ref{ass:stab} and \ref{ass:antiHur}, and is coupled with the state feedback controller in \eqref{eq:sf}. Then, there exist matrices $K$ and $L$ such that output regulation is achieved, i.e., $\lim_{t\to\infty} e(t)=0$ for all initial conditions and any piecewise continuous signal $r(t)$ and $(A+BK)$ is Hurwitz, if and only if the following conditions are true:
\end{thm} 
\begin{enumerate}
    \item\label{cond:1} There exist matrices $\Pi$ and $\Gamma$ that solve the linear matrix equations
    \begin{equation}\label{eq:reg}
    \begin{aligned}
    & \Pi S = A\Pi+B\Gamma+E,\\
    & 0 = C\Pi+F,
    \end{aligned}
    \end{equation}
    \item\label{cond:2} $\mathrm{im} (\Pi D)\subseteq \mathcal{V}(A,B,C)$.
\end{enumerate}

The first condition of Theorem \ref{theor:exactreg} characterizes the regulator equations \cite{huang2004}, which constitute necessary and sufficient conditions for achieving output regulation when $r(t)=0$. Criteria for their solvability and corresponding solution methods are detailed in \cite{huang2004}. 

\begin{rem}\label{rem:convenience}
Assumption \ref{ass:antiHur} is made only for convenience and is not necessary for the solution of the regulator equations. See  \cite{isidori2017} for a comprehensive discussion. 
\end{rem}

If the regulator equations are satisfied and the exosystem state gain in \eqref{eq:sf} is chosen as

\begin{equation}\label{eq:exo_gain}
    L = \Gamma-K\Pi,
\end{equation}
then, the closed-loop system in \eqref{eq:wsystem} can be rewritten in coordinate $\bar{x}=x-\Pi w$ as (omitting the exosystem dynamics)

\begin{equation}\label{eq:xbarsys}
\begin{aligned}
    & \dot{\bar{x}}=\left(A+BK\right)\bar{x}-\Pi Dr, \\
    & e=C\bar{x}.
\end{aligned}
\end{equation}

For this modified system, achieving output regulation requires solving a disturbance decoupling problem with internal stability \cite{saberi2001}, i.e., finding a matrix $K$ such that the transfer matrix from $r$ to $e$ is identically zero and $(A+BK)$ is Hurwitz. The second condition of Theorem \ref{theor:exactreg} establishes the criterion for the existence of such a $K$, and  \cite{wonham1985} outlines an algorithmic procedure to guide its computation.

\begin{rem}\label{rem:uniqueness}
If the matrix 
\begin{equation}\label{eq:sylv_cond}
   \begin{bmatrix}
       A-\lambda I & B \\
       C & 0
   \end{bmatrix} 
\end{equation}
is square (the number of inputs equals the number of outputs) and has full rank $\forall \lambda \in \sigma(S)$, where $\sigma(S)$ denotes the spectrum of $S$, then the solution to \eqref{eq:reg} is unique for any $E$ and $F$ \cite[Lemma~4.1]{isidori2017}, which in turn guarantees the uniqueness of the exosystem gain. For a minimal realization, i.e. $(A,C)$ is observable and $(A,B)$ is controllable, the values where the matrix in \eqref{eq:sylv_cond} is not full rank are the zeros of $C(Is-A)^{-1}B$ (transmission zeros). In this case, to achieve output regulation, the exosystem gain must be chosen as in \eqref{eq:exo_gain} \cite[Theorem~1.7]{huang2004}.    
\end{rem}

\section{Compliant Control Problem}\label{sec:ccp}

In tasks involving interaction between a system and an environment, it is often desirable to control the dynamic relationship between flow variables (typically position or velocity) and effort variables (typically force or torque) at the interaction port \cite{schumacher2019}. This section formalizes this objective, introduces the relevant variables and models, and states the exact compliant control problem under suitable assumptions.

Let $f_i$ denote the force exchanged between a system and an environment. To regulate this interaction, a desired impedance behavior is often specified by defining a target relationship between displacement error and interaction force. This relationship is typically expressed as a function $f_r\left(e,\dot{e},...,e^{(n)}\right)$, where $e=x_r-x_e$, with $x_r \in \mathbb{R}^e$ denoting the reference and $x_e \in \mathbb{R}^e$ the actual displacement at the interaction point. Although more general nonlinear  models have been proposed in the literature \cite{Shimizu2012}, $f_r$ is commonly specified as a second-order linear model
\begin{equation}\label{eq:ref_impedance}
    f_r = M_d\ddot{e}+B_d\dot{e}+K_de.
\end{equation}
The matrices $M_d$, $B_d$ and $K_d$ are often selected such that, after substituting $f_r$ by $f_i$ in \eqref{eq:ref_impedance}, the resulting ideal behavior operator defined by $H_1:f_i\mapsto \dot{e}$ is output strictly passive (OSP). 
For example, \cite{Haninger2022} chose $M_d$ and $B_d$ positive definite with $K_d=0$, whereas \cite{Rashad2022} chose $B_d$ and $K_d$ positive definite with $M_d=0$. 

Since $f_i$ arises from physical interaction, it is generally reasonable to assume that it depends on the system's motion at the interaction port -- that is, it can be expressed as a function of $x_e$ and possibly its time derivatives \cite{Hogan1984}. In many scenarios, $f_i$ may be represented as the output of a passive system $H_2$ with the velocity $\dot{x}_e$ at the interaction port as an input. Common spring-damper models of the environment are special cases of this class \cite{Dong2019,Yanghong2025}. As a consequence, the relation $f_i=f_r$ (depicted in Fig.~\ref{im:feedback}) can be viewed as the feedback interconnection of two systems: a passive system $H_2$, which generates $f_i$, and an OSP system $H_1$, associated with the target impedance relation $f_r$. In Figure~\ref{im:feedback}, $\dot{\bar{x}}_r$ represents a simple transformation on $\dot{x}_r$ to ensure the map from $\dot{\bar{x}}_r$ to $\dot{x}_e$ is also $H_1$. Due to the assumptions on the defined maps, the illustrated feedback interconnection guarantees that the closed-loop system is $\mathcal{L}_2$-finite gain stable \cite[Theorem~5.2]{brogliato2020}, and is one of the reasons that motivate defining the equality $f_r=f_i$ as an objective in compliant control strategies. This key assumption is stated formally below.

\begin{figure}[h!]
  \centering
  \includegraphics[scale = 0.8]{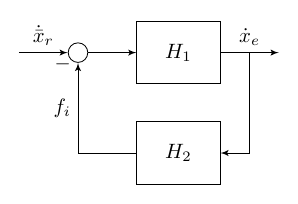}
\caption{Feedback connection illustrating  the impedance relationship.}\label{im:feedback}
\end{figure}

\begin{assum}\label{ass:OSP}
The map $H_1:-f_i\mapsto \dot{x}_e$ is OSP, and the map $H_2:\dot{x}_e\mapsto f_i$ is passive.
\end{assum}

\begin{rem}\label{rem:boundedxe}
    For the second order impedance model \eqref{eq:ref_impedance} Assumption \ref{ass:OSP} implies $B_d\succ0$ and $M_d,K_d\succeq0$. For $K_d\succ0$, one can further prove using a composite Lyapunov function (the mechanical energy for $H_1$ and the storage function derived from the passivity hypothesis for $H_2$) and boundedness theorems \cite[Theorem 8.4]{khalil2002} that the state $x_e$ is bounded. For $K_d\succeq0$, the state $x_e$ might not be bounded (consider, for instance, a single degree of freedom system and second order impedance model, with $k_d=0$ and $f_i=b\dot{x}_e$).
\end{rem}

\begin{rem}\label{rem:passiveimpedance}
    In some circumstances, a pure stiffness is selected for the impedance model \cite{Calanca2017}, which violates the OSP condition. Under this configuration, the relationship $f_i=f_r$, with $f_i$ passive, results in a Lyapunov stable system in the absence of a reference signal $x_r$, but may result in an $\mathcal{L}$-unstable (in the input-output sense) system otherwise. 
\end{rem}

Assume additionally that the system linearized model is given by \eqref{eq:system}, where the product $Ew$ accounts for the effects of the interaction force on the plant, i.e., $w$ may include components beyond $f_i$, but their influence on the plant is canceled by the particular structure imposed on the matrix $E$. Furthermore, the connection between the interaction port motion variable and the system state can be established by a suitable linear function.

With these preliminaries established, we can now formally state the exact compliant control problem.

\begin{prob}[Exact compliant control]\label{prob:exact}
Consider the system described by \eqref{eq:system}, where $Ew$ denotes the contribution of the interaction force $f_i$ to the plant dynamics. The reference impedance model is given in \eqref{eq:ref_impedance}, and the interaction port motion variable is assumed to be a linear function of the system state $x(t)$. Under assumptions \ref{ass:stab} and \ref{ass:OSP}, the exact compliant control problem consists in determining a control action $u$ such that 
\begin{enumerate}
    \item\label{cond:ex1} $\lim_{t\to\infty} \left(f_r(t)-f_i(t)\right)=0$ for every sufficiently smooth and bounded signal $x_r(t)$;
    \item\label{cond:ex2} $x(t)$ remains bounded for all $t\ge0$.
\end{enumerate}
\end{prob}

Some comments about the definition of Problem~\ref{prob:exact} are in order. In its conventional form \cite{Valency2003,Ott2015}, the compliant control problem is posed as finding a control action such that
\begin{equation}\label{eq:trad}
f_i=M_d\ddot{e}+B_d\dot{e}+K_de
\end{equation}
is satisfied at the interaction port, without prescribing the exact time domain of validity. In the definition of Problem \ref{prob:exact}, condition \ref{cond:ex1} requires this identity to hold asymptotically, which allows an equivalence between the compliant control problem and the standard output regulation formulation. 

Assumption \ref{ass:stab} is introduced to guarantee the existence of a controller that achieves exact compliant control, while Assumption  \ref{ass:OSP} is defined to guarantee a stable compliant behavior at the interaction port. Although this behavior is typically regarded as sufficient to keep the system state bounded \cite{Zou2022}, condition \ref{cond:ex2} explicitly enforces this property.

\section{Compliant control as a generalized output regulation problem}

This section reformulates the exact compliant control problem within the framework of generalized output regulation and, assuming static feedback from a specified measurement set, derives the control law structure that achieves exact regulation. Due to some particularities regarding the inclusion of viscoelasticity between the actuator and its load, we address the cases of stiff and soft (series elastic-damper) joints \cite{Calanca2016} separately. To highlight key aspects of the solution, we focus on the cases of single (stiff) and two (soft) degrees of freedom. By focusing on minimal yet representative dynamics, these models make it easier to establish direct links between design choices and system behavior \cite{Valency2003,Calanca2018,Keemink2018}. 

\subsection{Stiff Joints}

\begin{figure}[h!]
\centering
  \includegraphics[scale = 0.8]{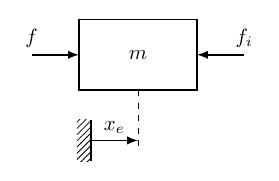}
\caption{Schematic representation of a stiff joint.}\label{im:stiff}
\end{figure}%

Figure \ref{im:stiff} shows a representative model of a stiff joint. Here, the term stiff denotes a connection between the actuator and its load that is rigid enough so that the entire system can be modeled as a lumped inertia $m$. The force $f$ is the control force and $f_i$ is the interaction force. The plant dynamics,
\begin{equation}\label{eq:stiff}
    m\ddot{x}_e=f-f_i,
\end{equation}
may be equivalently expressed in the state-space representation \eqref{eq:system}. Defining the state
\begin{equation}\label{eq:st_def}
    x= \left[x_e \,\, \dot{x}_e\right]^\intercal
\end{equation} 
and the control input $u=f$ leads to the system matrices
\begin{equation}\label{eq:st_sys}
    A = \begin{bmatrix}
        0 & 1 \\
        0 & 0
    \end{bmatrix}, \quad B=\left[\begin{array}{c}
    0 \\
    1/m  \end{array}\right].
\end{equation}
The explicit form of $E$, representing the effects of the interaction force on the plant, is given after the definition of the exosystem state in \eqref{eq:exo_state}.

Given this model, the design of a control input that ensures exact compliant regulation does not naturally align with conventional control problems (e.g. output tracking, disturbance rejection, noise attenuation). However, it can be reformulated within the output regulation framework by exploiting the relation $f_i=f_r$ to construct an auxiliary system that generates a reference trajectory for the inertia $m$. This generating reference strategy, known as position-based impedance control or admittance control \cite{schumacher2019}, is analogous to other control strategies employing auxiliary systems, but here the system generates the reference trajectory rather than estimating parameters (adaptive control) or reconstructing states (observers). The main idea is to replace the interaction port displacement in \eqref{eq:trad} by a reference generator $\hat{x}$. For the stiff joint, this substitution leads to
\begin{equation}\label{eq:aux}
    f_i=m_d(\ddot{x}_r-\ddot{\hat{x}})+b_d(\dot{x}_r-\dot{\hat{x}})+k_d(x_r-\hat{x}),
\end{equation}
where $\hat{x}$ represents the reference trajectory for $x_e$. Lowercase symbols are used for the impedance parameters to highlight that these quantities are scalars for the stiff system. 

\begin{lem}\label{lem:equi}
    Define 
    \begin{equation}
        \hat{e}=\hat{x}-x_e,\label{eq:ehat}
    \end{equation}
    where $\hat{x}$ is given by \eqref{eq:aux}. Assume that $\dot{\hat{e}}$ and $\ddot{\hat{e}}$ are uniformly continuous. Then, under Assumption \ref{ass:OSP}, the following holds:
    \begin{enumerate}
        \item If $k_d \ne0$, then
        \begin{equation}
            \lim_{t\to\infty} \left(f_r(t)-f_i(t)\right)=0 \Longleftrightarrow\lim_{t\to\infty} \hat{e}=0. 
        \end{equation}
        \item If $k_d=0$, then
        \begin{equation}
            \lim_{t\to\infty} \left(f_r(t)-f_i(t)\right)=0 \Longleftrightarrow\lim_{t\to\infty} \dot{\hat{e}}=0. 
        \end{equation}        
    \end{enumerate}
\end{lem}

\begin{pf}
    See Appendix \ref{app:lem1}. 
\end{pf}

Lemma \ref{lem:equi} specifies conditions under which Problem \ref{prob:exact} can be equivalently reformulated by replacing the convergence condition on the force error with the corresponding condition on position or velocity error. Instead of requiring the displacement and its derivatives to vanish, we assume the derivatives are uniformly continuous. This choice ensures that the number of regulated outputs does not exceed the number of inputs, which in turn ensures the existence of a solution to the regulator equations. To finalize its reformulation within the output regulation structure, it remains to specify the exosystem \eqref{eq:exosystem}
and error \eqref{eq:error} equations.

For a more general derivation, we consider the case where $k_d\ne0$, and define the error equation \eqref{eq:error} as in \eqref{eq:ehat}, i.e., $e=\hat{e}$. 
Comments on the particularities of the solution when $k_d=0$ are given in Remark \ref{rem:kd0}.

Because \eqref{eq:aux} determines $\hat{x}$ independently of $x_e$, it is appropriate to define $\hat{x}$ as a state of the exosystem. In this formulation, writing the full dynamics \eqref{eq:aux} for $\hat{x}$ (with parameters $m_d,b_d,k_d$), and the interaction model $f_i$ inside the exosystem creates two problems: 1) it produces algebraic cancellations in the derivation of the control law that obfuscate the regulator design and 2) it introduces an explicit dependence of the exosystem on the plant state when $f_i$ depends on $x_e$.  Instead, and without loss of generality for our synthesis, we treat $\hat{x}$ and $f_i$ as bounded, sufficiently smooth signals and represent their time derivatives as the exosystem inputs. This formulation preserves plant–exosystem separation, and provides a general controller derivation that remains valid even for nonlinear models of the interaction force. The justification rests on the fact that convergence of $x_e$ to $\hat{x}$, combined with Assumption \ref{ass:OSP}, entails boundedness of both $\hat{x}$ and $f_i$. 

Finally, as the controller structure \eqref{eq:sf} assumes static state feedback, the definition of the exosystem states directly determines the feedback signals. Thus, to investigate which signals are needed for output regulation, we include $\hat{x}$, $f_i$ and its derivatives of order $a$ and $b$ in the exosystem state
\begin{equation}\label{eq:exo_state}
    w = \left[\begin{array}{c c c c c c c c}
    \hat{x} & \dot{\hat{x}} & \ldots & \hat{x}^{(a)} & f_i & \dot{f}_i & \ldots & f_i^{(b)} \end{array}\right]^\intercal
\end{equation}
and let the terms of order $a+1$ and $b+1$ be the inputs of the exosystem
\begin{equation}
    r=\left[\begin{array}{c c}
    \hat{x}^{(a+1)} & f_i^{(b+1)} \end{array}\right]^\intercal.
\end{equation}

With this definition, the model formulation is completed, enabling the specification of the remaining matrices: the matrix $E\in \mathbb{R}^{2 \times (a+b)}$ that completes the plant dynamics (Eq. \eqref{eq:stiff})
\begin{equation}\label{eq:st_sys2}
    E_{ij}=\begin{cases}
    -\frac{1}{m}, & \text{if $(i,j)=(a+1,2)$}.\\
    0, & \text{otherwise}.
  \end{cases} 
\end{equation}
where $(\cdot)_{ij}$ denotes the $(i,j)$-th entry of the corresponding matrix; the exosystem matrices $S\in \mathbb{R}^{(a+b) \times (a+b)}$ and $D\in \mathbb{R}^{(a+b) \times 2}$ (whose structure follows directly from $w$ including sequences of increasing-order derivatives among its components)
\begin{subequations}\label{eq:st_exo}
\begin{gather}
  S_{ij}=\begin{cases}
    1, & \text{if $j=i+1 \,\,\textrm{and} \,\, i\notin\{a,a+b\} $}.\\
    0, & \text{otherwise}.
  \end{cases} \\
  D_{ij}=\begin{cases}
    1, & \text{if $(i,j) \in\{(a,1),(a+b,2)\}$}.\\
    0, & \text{otherwise}.
  \end{cases}\label{eq:D} 
  \end{gather}
\end{subequations}
and the error matrices $C\in \mathbb{R}^{1 \times 2}$ and $F\in \mathbb{R}^{1 \times (a+b)}$ (Eq.~\eqref{eq:ehat})
\begin{subequations}\label{eq:st_error}
\begin{gather}
   C=\left[\begin{array}{c c}
    -1 & 0 \end{array}\right], \\
  F_{j}=\begin{cases}
    1, & \text{if $j=1$}.\\
    0, & \text{otherwise}.
  \end{cases} 
  \end{gather}
\end{subequations}

\begin{lem}\label{lem:reg_sol}
    Consider the system given in \eqref{eq:wsystem}, with the corresponding matrices specified in \eqref{eq:st_sys}, \eqref{eq:st_sys2}, \eqref{eq:st_exo} and \eqref{eq:st_error}. Assuming state feedback \eqref{eq:sf}, output regulation is achieved if and only if 
    \begin{equation}\label{eq:sol}    u=k_1\hat{e}+k_2\dot{\hat{e}}+f_i+m\ddot{\hat{x}},        
    \end{equation} 
    with $\hat{e}$ defined in \eqref{eq:ehat} and $k_1,k_2>0$.
\end{lem}
\begin{pf}
    Since Theorem \ref{theor:exactreg} provides necessary and sufficient conditions for output regulation, the proof reduces to verifying conditions \ref{cond:1} and \ref{cond:2} of that theorem, together with the uniqueness and selection requirements for the exosystem gain given in Remark~\ref{rem:uniqueness}.
    
    To check condition \ref{cond:1}, note that it follows from \eqref{eq:st_sys} that $(A,B)$ is controllable, and thus Assumption~\ref{ass:stab} holds. In addition, since $(A,C)$ is observable and $C(Is-A)^{-1}B$ has no transmission zeros, Remark~\ref{rem:uniqueness} guarantees that the solution of \eqref{eq:reg} is unique. A direct calculation shows that the solution, $\Pi\in \mathbb{R}^{2 \times (a+b)}$ and $\Gamma\in \mathbb{R}^{1 \times (a+b)}$, is 
\begin{subequations}\label{eq:reg_sol}
\begin{gather}
  \Pi_{ij}=\begin{cases}
    1, & \text{if $i=j$.}\\
    0, & \text{otherwise}.
  \end{cases}\label{eq:Pi} \\
  \Gamma_{j}=\begin{cases}
    m, & \text{if $j=3 $}.\\
    1, & \text{if $j=a+1 $}.\\
    0, & \text{otherwise}.
  \end{cases} 
  \end{gather}
\end{subequations}
Condition \ref{cond:2} specifies the criterion for the solvability of the disturbance decoupling problem. However, since \eqref{eq:D} and \eqref{eq:Pi} imply $\Pi D=0$, there is no disturbance for the resulting system \eqref{eq:xbarsys}, and the disturbance decoupling problem reduces to finding 
\begin{equation}\label{eq:stiffK}
    K=-[k_1\,\,k_2]
\end{equation}
that renders $(A+BK)$ Hurwitz. One readily verifies that this is true for any $k_1>0$ and $k_2>0$.
Replacing \eqref{eq:reg_sol} and \eqref{eq:stiffK} into \eqref{eq:exo_gain} gives
\begin{equation}\label{eq:stiffL}
    L = \left[\begin{array}{c c c c c c}
    k_1 & k_2 & m & 0_{a-3} & 1 & 0_{b-1}   \\
    \end{array}\right],
\end{equation}
where $0_k$ denotes a zero vector of length $k$. Substituting \eqref{eq:stiffK} and \eqref{eq:stiffL} into \eqref{eq:sf} yields \eqref{eq:sol}. 
\end{pf}

\begin{rem}\label{rem:kd0}
    For $k_d=0$, it is appropriate to set the error in \eqref{eq:error} to $e=\dot{x}_e-\dot{\hat{x}}$, as established in Lemma~\ref{lem:equi}. By a similar argument, one can redefine the state in \eqref{eq:st_def} eliminating the displacement $x_e$ and keeping only $\dot{x}_e$. Similarly, the state $\hat{x}$ can also be eliminated from $w$ in \eqref{eq:exo_state}. With this procedure and doing the same calculations as in Lemma~\ref{lem:reg_sol}, the resulting control law is also \eqref{eq:sol}, but with $k_1=0$.
\end{rem}

Lemma~\ref{lem:reg_sol} shows that, for the stiff system, asymptotic tracking of a reference signal using static feedback on the measurements $x$ and $w$ is achieved if and only if the control law \eqref{eq:sol} is applied. As a consequence of the construction of the exosystem, this holds not only for signals produced by an auxiliary generator such as \eqref{eq:aux} but for arbitrary smooth reference trajectories. The particularity of a reference generated by \eqref{eq:aux} is that, \textit{a priori}, boundedness of $\hat{x}$ and $f_i$ is not guaranteed (e.g., for $f_i=-kx_e,$ with $k>k_d$). Assumption \ref{ass:OSP} provides sufficient condition for all variables in the control law \eqref{eq:sol} to remain bounded and prevent the described situation.

With those results established, we can now state the main result of this section.

\begin{thm}\label{theor:stiffsol}
    Consider the stiff system in \eqref{eq:stiff} subject to static state feedback \eqref{eq:sf}, where $x$ and $w$ are defined in \eqref{eq:st_def} and \eqref{eq:exo_state}, respectively. Let \eqref{eq:aux} describe the dynamics of $\hat{x}$, with the corresponding parameters chosen such that Assumption \ref{ass:OSP} holds. Under these conditions, the exact compliant control problem is solved if and only if the control law \eqref{eq:sol} is employed, with $k_1\ge0$ and $k_2>0$. 
\end{thm}
\begin{pf}
    Equations \eqref{eq:st_def}, \eqref{eq:st_sys}, \eqref{eq:exo_state} and \eqref{eq:st_sys2} define the system in form \eqref{eq:system}, with $Ew$ denoting the influence of $f_i$ on the plant dynamics. Lemma~\ref{lem:equi} establishes that, provided $\dot{\hat{e}}$ and $\ddot{\hat{e}}$ are uniformly continuous, achieving zero force error is equivalent to achieving either zero position or velocity error, depending on whether $k_d$ is zero. Considering this equivalence, Lemma~\ref{lem:reg_sol} specifies the unique structure of the controller \eqref{eq:sol} under static state feedback that achieves output regulation. The design requires $k_1,k_2>0$, although, by Remark~\ref{rem:kd0}, the coefficient $k_1$ may be set to zero when $k_d=0$. Since the substitution of \eqref{eq:sol} in \eqref{eq:stiff} leads to an exponentially stable system, the uniform continuity hypothesis required by Lemma \ref{lem:equi} is satisfied. In Lemma~\ref{lem:reg_sol}, the proof presupposes that the input signals are bounded. Assumption \ref{ass:OSP} provides sufficient conditions to ensure that those variables are indeed bounded, thus verifying condition~\ref{cond:ex1} of Problem~\ref{prob:exact}. Since the interaction port motion variable coincides with the system state, condition~\ref{cond:ex2} is automatically fulfilled.
\end{pf}

Theorem \ref{theor:stiffsol} establishes the unique control law structure that solves the exact compliant control problem under the assumption of static state feedback of the signals specified in \eqref{eq:st_def} and \eqref{eq:exo_state}. Although control law \eqref{eq:sol} can also be obtained using Lyapunov-based design methods, the generalized output regulation framework allows proving its uniqueness. 

As a consequence of Theorem \ref{theor:stiffsol}, we can conclude that although the conventional PD controller can render a passive mapping from $f_i$ to $\dot{x}_e$ \cite{Keemink2018}, it is not capable of realizing the target relationship \eqref{eq:trad}, even with the inclusion of force compensation. In robotics terminology \cite{Calanca2016}, this means that impedance rendering with such a controller always produces an error (except in rare cases in which the term $f_i+m\ddot{\hat{x}}$ in eq.~\eqref{eq:sol} vanishes, as discussed in Section~\ref{sec:num}). 

Similar to the impedance control strategy with acceleration feedback \cite{Calanca2017}, the control law \eqref{eq:sol} ensures that the rendered impedance exactly matches the desired impedance for any desired parameters across the entire frequency spectrum. In practice, the inclusion of unmodeled effects (time delay, noise, friction, actuator dynamics) may compromise this ideal behavior, leading to deviations from the desired dynamics. 

The next section demonstrates how the control framework developed here can be extended to achieve exact compliant control in soft joints.

\subsection{Soft Joints}

The designation \textit{soft joint} applies to systems in which elastic or damping properties are deliberately introduced between the actuator and its load. Some benefits of this configuration include enhanced shock tolerance, improved force control robustness, and low-cost force measurement \cite{Calanca2017}. 

\begin{figure}[h!]
  \centering
  \includegraphics[scale = 0.8]{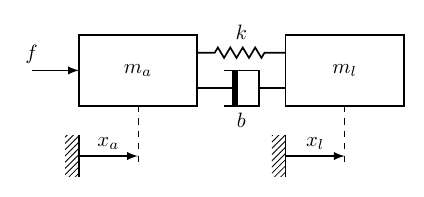}
\caption{Schematic representation of a soft joint.}\label{im:soft}
\end{figure}

A schematic representation of a soft joint is depicted in Figure~\ref{im:soft}. The subscript $a$ corresponds to actuator variables, while $l$ denotes load variables. The inertias representing these two elements are connected by a spring and a damper with known parameters $k$ and $b$. Although damping is often neglected in such configurations \cite{Calanca2016}, it is included here for generality. The interaction port in this system is characterized by the load motion variable $x_l$ and the viscoelastic force
\begin{equation}\label{eq:visc}
    f_i=k(x_a-x_l)+b(\dot{x}_a-\dot{x}_l).
\end{equation}
Accordingly, the reference model is \eqref{eq:ref_impedance}, with the error defined by 
\begin{equation}\label{eq:soft_imp_er}
e=x_r-x_l.    
\end{equation}

A strategy like the one used to derive \eqref{eq:aux} for the stiff system can also turn the force-tracking problem into a position-tracking one. However, in the soft-joint case, replacing the interaction displacement in \eqref{eq:trad} (where $e$ is defined in \eqref{eq:soft_imp_er}) with the reference generator $\hat{x}$ yields a reference for the load position, which leads to a non-collocated control problem \cite{Cannon1984,Cruzneto2019}.
Although this problem admits a solution, the resulting control law is considerably more complex. Following a procedure analogous to that of the previous section (proof of Lemma \ref{lem:reg_sol}), one can show that the resulting control law depends on the third derivative of the reference signal as well as on the inertia properties of the load, which are often unavailable in practice. 
To avoid the non-collocation problem for a soft joint represented only by a spring, \cite{Pratt2004} proposed an alternative scheme that results in a collocated control problem. However, as noted in \cite{Calanca2017}, this scheme only approximates the desired behavior and does not fully realize the exact target impedance.

Instead, we can have a collocated problem that renders the target impedance exactly by leveraging the known soft joint parameters $k$ and $b$. We construct a reference for the actuator displacement by substituting the interaction force model \eqref{eq:visc} into $f_i=f_r$ and replacing the actuator displacement with the reference generator $\hat{x}$

\begin{equation}\label{eq:soft_gen}
    k(\hat{x}-x_l)+b(\dot{\hat{x}}-\dot{x}_l)=m_d(\ddot{x}_r-\ddot{x}_l)+b_d(\dot{x}_r-\dot{x}_l)+k_d(x_r-x_l).
\end{equation}

Differently from traditional admittance implementations, \eqref{eq:soft_gen} generates the reference position using only displacement signals and their derivatives as inputs, rather than force inputs. For a general impedance model, implementing \eqref{eq:soft_gen} requires access to the load acceleration, which is often sensitive to noise. In contrast, for reduced models such as the Voigt model ($m_d=0$), this issue can be avoided. The advantage of implementing \eqref{eq:soft_gen} to generate the reference trajectory is that it enables the direct use of the control design developed in the previous section for the stiff joint. This is possible because the actuator dynamics remain unchanged, as the actuator inertia is still directly driven by the control ($f$) and the interaction ($f_i$) forces, the latter being known \eqref{eq:visc}. 

However, to guarantee exact compliant control, it is necessary to perform an additional step to verify condition \ref{cond:ex2} of Problem \ref{prob:exact}. This condition guarantees the boundedness of the full state, which is necessary because the viscoelastic element between the actuator and the load increases the order of the system. To this end, consider that the actuator and load dynamics are represented by

\begin{subequations}\label{eq:soft_dyn}
\begin{gather}
  m_a\ddot{x}_a=f-f_i, \\
  m_l\ddot{x}_l=f_i,\label{eq:load_dyn}
  \end{gather}
\end{subequations}
with $f_i$ given by \eqref{eq:visc}.
To perform a stability analysis, the state of the complete system is more conveniently defined as
\begin{equation}\label{eq:scomplete}
    x=\left[\begin{array}{c c c c c}
    \hat{e} & \dot{\hat{e}} & x_l & \dot{x}_l & \hat{x}
    \end{array}\right]^\intercal.
\end{equation}
where 
\begin{equation}\label{eq:soft_er}
\hat{e}=\hat{x}-x_a
\end{equation}
Because the control law \eqref{eq:sol} depends on $\ddot{\hat{x}}$, generating $\hat{x}$ via \eqref{eq:soft_gen} would yield a dependence on the third derivative of the load displacement in the control input. To avoid introducing higher-order derivatives of the load motion, the state $\hat{x}$ in \eqref{eq:scomplete} is defined according to \eqref{eq:soft_gen} with $m_d=0$. 

\begin{prop}\label{prop:bounded_state}
    Consider the soft system in \eqref{eq:soft_dyn}, where $f$ is given by the control law \eqref{eq:sol}, with $\hat{e}$ given by \eqref{eq:soft_er}, and $\hat{x}$ generated according to \eqref{eq:soft_gen} with $m_d=0$. Then, for all initial conditions and any sufficiently smooth and bounded signal $x_r(t)$, the state of the complete system $x(t)$ in \eqref{eq:scomplete} remains bounded for all $t\ge0$. 
\end{prop}
\begin{pf}
    See Appendix \ref{app:prop1}.
\end{pf}

As a consequence of Proposition \ref{prop:bounded_state} and Theorem \ref{theor:stiffsol}, we have the following corollary.

\begin{cor}\label{cor}
    Consider the soft system \eqref{eq:soft_dyn} and the reference impedance model \eqref{eq:ref_impedance}, with $M_d=0$ and $e$ given by \eqref{eq:soft_imp_er}. Under the control law \eqref{eq:sol} with $k_1,k_2>0$, using $\hat{e}$ from \eqref{eq:soft_er} and $\hat{x}$ generated by \eqref{eq:soft_gen}, the exact compliant control problem is solved. 
\end{cor}

Corollary \ref{cor} establishes that generating the reference trajectory using \eqref{eq:soft_gen} is sufficient to solve the exact compliant control problem with a control law that does not rely on the load parameters. By a short additional argument, using the ideas in the proof of Proposition \ref{prop:bounded_state}, the same conclusion holds when the load is subject to additional bounded forces. Note that if the dynamics of the viscoelastic element are of lower order than those of the impedance model, the controller may require higher-order derivatives of the load displacement (i.e., derivatives beyond the load acceleration).

\section{Numerical Example}\label{sec:num}

This section illustrates, through a numerical example, the differences in performance between the control law \eqref{eq:sol}, hereafter referred to as ECC (exact compliant controller), and a commonly employed PD controller.

We first consider the stiff system \eqref{eq:stiff} with the environment modeled as a linear spring
\begin{equation}\label{eq:rigidenv}
    f_i=k_ex_e.
\end{equation}
A periodic function is adopted for the reference signal to evaluate performance across varying excitation frequencies: 
\begin{equation}
    x_r=A\textrm{sin}(2\pi f t).
\end{equation}
Table~\ref{tab:param} summarizes the joint parameters, target impedance, control gains, and the reference amplitude $A$ used in the simulations. The gains $k_1$ and $k_2$ are the same for both controllers. Choosing high PD gains is a common strategy in admittance control \cite{Ott2015}.

\begin{table}[h!]
\caption{Simulation parameters for the stiff system.}
\centering
\begin{adjustbox}{width=\columnwidth}
\begin{tabular}{c | c c c c c c c}
\hline
Parameter & $m$ & $m_d$ & $b_d$ & $k_d$ & $A$ & $k_1$ & $k_2$ \\ \hline
Value & 2 & 1 & 10 & 100 & 0.2 & $10^4$ & $10^3$ \\ \hline
Unit & kg & kg & Ns/m & N/m & m & N/m & Ns/m \\
\hline
\end{tabular}
\end{adjustbox}
\label{tab:param}
\end{table}

Target impedance rendering performance can be analyzed using several different metrics \cite{Santos2025,Zanette2025}. Here, the root mean square error (RMSE) between $f_r$ and $f_i$ in steady-state is adopted. This index was evaluated for reference signal excitation frequencies $f\in \left[0.1,\,10\right]$ Hz, and environmental stiffness $k_e\in \left[1, \,9500\right]$ N/m. The upper bound on $k_e$ was chosen because the PD controller becomes unstable for larger stiffness values.

Two analyses were performed: (i) a nominal system with known parameters, and (ii) an uncertain system including additive Gaussian noise (zero mean and standard deviation of 0.1) in the interaction force measurement and Stribeck friction. The results of both analyses are presented as color-map plots in Fig.~\ref{im:stiffres}.

\begin{figure}[t]
    \centering
    \begin{subfigure}{0.49\linewidth}
        \centering
        \includegraphics[width=\linewidth]{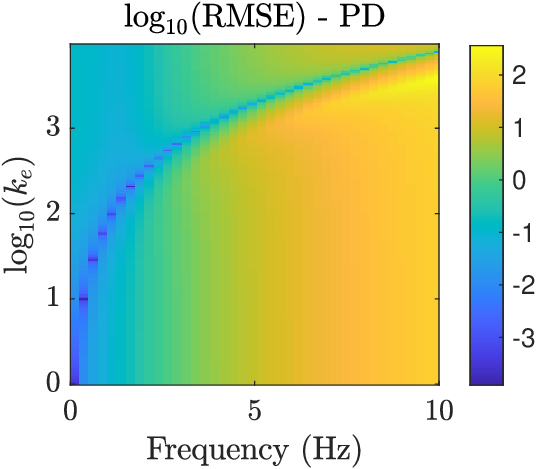}
        \caption{}\label{im:stiff_nom_pd}
    \end{subfigure}
    \hfill
    \begin{subfigure}{0.49\linewidth}
        \centering
        \includegraphics[width=\linewidth]{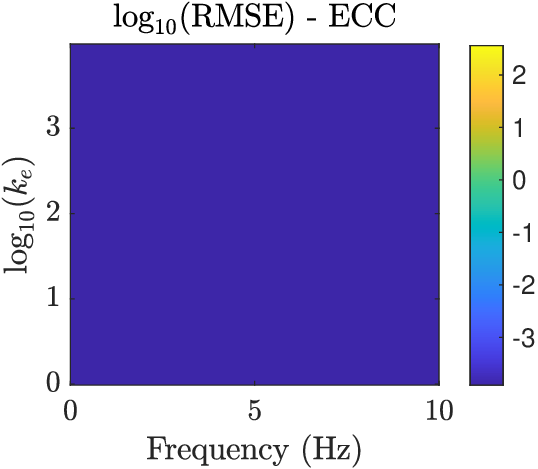}
        \caption{}\label{im:stiff_nom_ecc}
    \end{subfigure}
    \medskip
    \begin{subfigure}{0.49\linewidth}
        \centering
        \includegraphics[width=\linewidth]{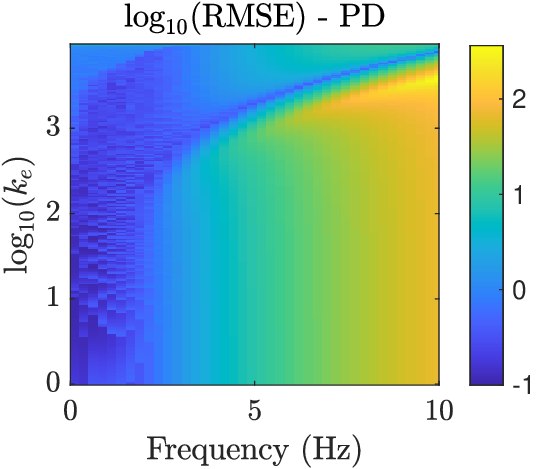}
        \caption{}\label{im:stiff_unc_pd}
    \end{subfigure}
    \hfill
    \begin{subfigure}{0.49\linewidth}
        \centering
        \includegraphics[width=\linewidth]{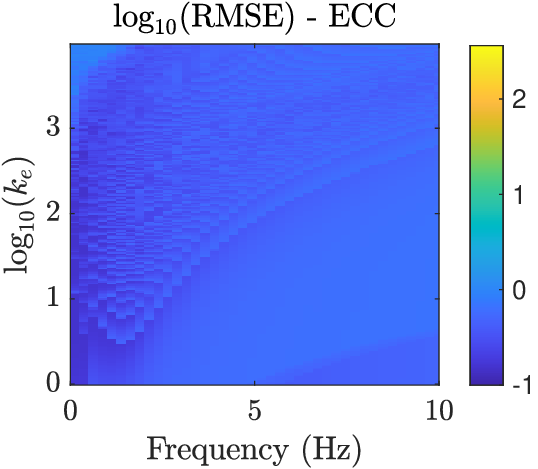}
        \caption{}\label{im:stiff_unc_ecc}
    \end{subfigure}
    \caption{RMSE of $(f_r-f_i)$ for the stiff joint: (a), (b) nominal model; (c), (d) Stribeck friction and measurement noise.}
    \label{im:stiffres}
\end{figure}

Figures \ref{im:stiff_nom_pd} and \ref{im:stiff_nom_ecc} show the results for the nominal system. As a numerical confirmation of Theorem~\ref{theor:stiffsol}, the ECC controller yields zero tracking error across all excitation frequencies and environmental stiffnesses (residual errors due to numerical integration are on the order of $10^{-7}$). For the PD controller, performance generally degrades as both the excitation frequency and the environmental stiffness increase. An exception appears along a curve where the RMSE vanishes, corresponding to the frequency $\omega=\sqrt{k_e/m}$ (only approximately visible due to the finite mesh resolution). At this frequency, the transfer functions from $x_r$ and its derivatives to $\hat{e}$ exhibit a zero (see Appendix \ref{app:trans}), explaining the observed null error. Since \eqref{eq:sol} is the unique control law that attains exact compliant control, the term $f_i + m\ddot{\hat{x}}$
in \eqref{eq:sol} must vanish at the frequency where the RMSE is zero; consequently, the ECC and PD controllers coincide at that frequency. This behavior is illustrated in Fig.~\ref{im:mx_fi_0}. The left subplot \ref{im:elip_mxfi} shows the system trajectory projected onto the $e\times F$ plane, yielding the characteristic elliptical plot and indicating convergence of $f_i$ to $f_r$. The right subplot \ref{im:0mxfi} shows the control term $f_i + m\ddot{\hat{x}}$ converging to zero, confirming the equivalence between ECC and PD controllers at this particular frequency.

\begin{figure}[t]
    \centering
    \begin{subfigure}{0.49\linewidth}
        \centering
        \includegraphics[width=\linewidth]{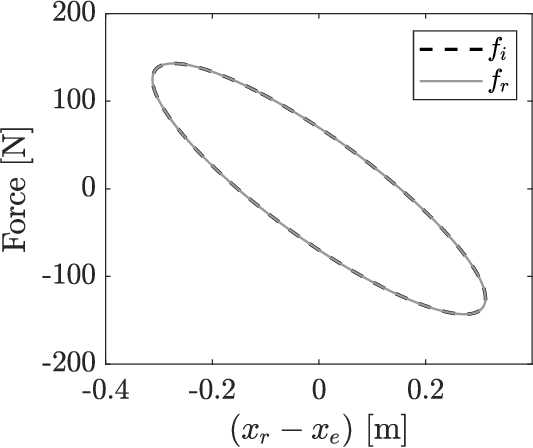}
        \caption{}\label{im:elip_mxfi}
    \end{subfigure}
    \hfill
    \begin{subfigure}{0.49\linewidth}
        \centering
        \includegraphics[width=\linewidth]{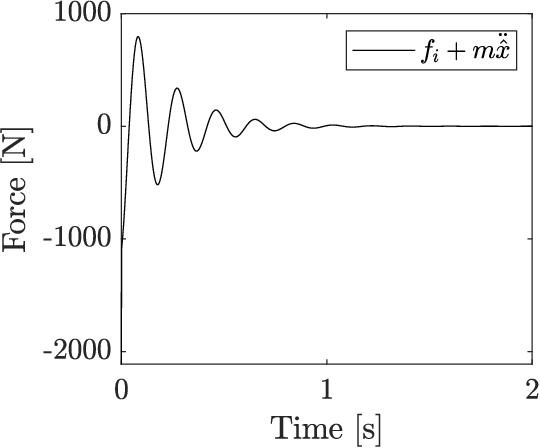}
        \caption{}\label{im:0mxfi}
    \end{subfigure}
    \caption{Results for the ECC controller considering $f=3.56$ Hz and $k_e=1000$ N/m ($\omega=\sqrt{k_e/m})$. \ref{im:elip_mxfi} shows the projection of the system trajectory on the $e\times F$ plane while \ref{im:0mxfi} shows the control term $f_i + m\ddot{\hat{x}}$.}
    \label{im:mx_fi_0}
\end{figure}

Figures \ref{im:stiff_unc_pd} and \ref{im:stiff_unc_ecc} present the results with Stribeck friction and force measurement noise. Although the ECC controller no longer attains zero tracking error, it still outperforms the PD controller for most combinations of environmental stiffness $k_e$ and excitation frequency $f$.

An analogous analysis was performed for the soft joint \ref{im:soft}. For a proper comparison with the results of the stiff system, the parameters were defined similarly (Table~\ref{tab:paramsoft}), and an additional external force, modeling a linear spring $f_{\mathrm{ext}} = k_e x_l$, was applied to the load. 

\begin{table}[h!]
\caption{Simulation parameters for the soft system.}
\centering
\begin{adjustbox}{width=\columnwidth}
\begin{tabular}{c | c c c c c c c c c}
\hline
Parameter & $m_a$ & $m_l$ & $b_d$ & $k_d$ & $A$ & $k_1$ & $k_2$ & $k$ & $b$\\ \hline
Value & 1 & 1 & 10 & 100 & 0.2 & $10^4$ & $10^3$ & $5000$ & $500$\\ \hline
Unit & kg & kg & Ns/m & N/m & m & N/m & Ns/m & N/m & Ns/m\\
\hline
\end{tabular}
\end{adjustbox}
\label{tab:paramsoft}
\end{table}

\begin{figure}[h!]
    \centering
    \begin{subfigure}{0.49\linewidth}
        \centering
        \includegraphics[width=\linewidth]{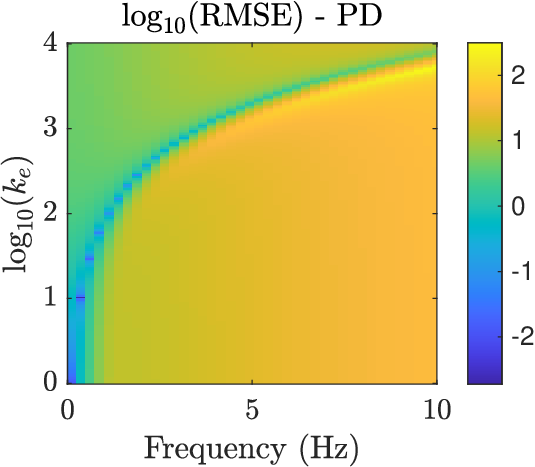}
        \caption{}
    \end{subfigure}
    \hfill
    \begin{subfigure}{0.49\linewidth}
        \centering
        \includegraphics[width=\linewidth]{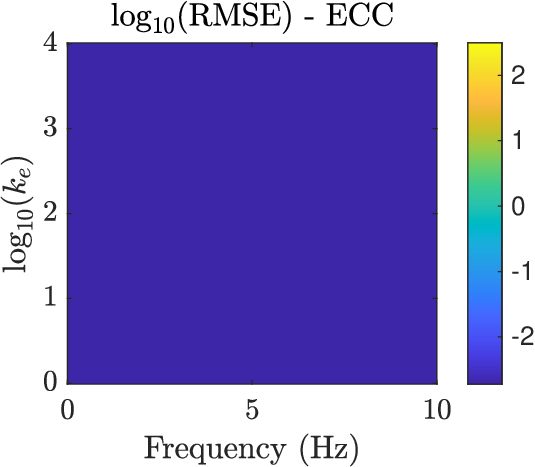}
        \caption{}
    \end{subfigure}
    \medskip
    \begin{subfigure}{0.49\linewidth}
        \centering
        \includegraphics[width=\linewidth]{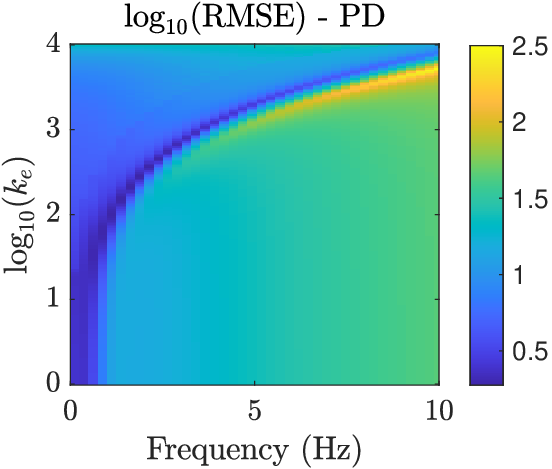}
        \caption{}
    \end{subfigure}
    \hfill
    \begin{subfigure}{0.49\linewidth}
        \centering
        \includegraphics[width=\linewidth]{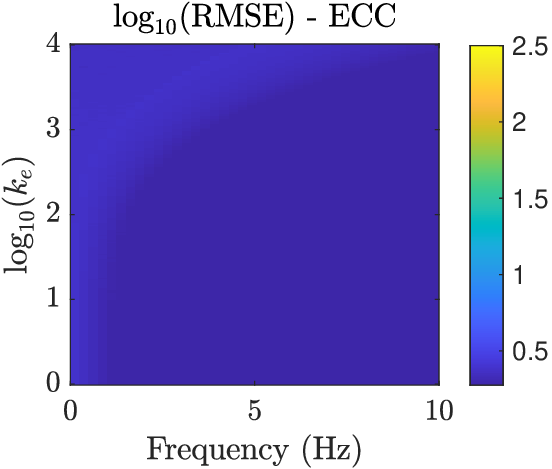}
        \caption{}
    \end{subfigure}
    \caption{RMSE of $(f_r-f_i)$ for the soft joint: (a), (b) nominal model; (c), (d) Stribeck friction and measurement noise.}
    \label{im:softres}
\end{figure}

Figure \ref{im:softres} presents the results for the soft system over varying excitation frequencies and environmental stiffnesses, considering both the nominal case and the case with measurement noise and Stribeck friction. Unlike the stiff joint case, the PD controller tolerated larger values of $k_e$ without becoming unstable, allowing the stiffness upper bound to be extended to $10^4$ N/m.

In general, the soft joint behavior is qualitatively similar to that of the stiff joint case. The PD controller again exhibits a zero-RMSE locus given by $\omega=\sqrt{k_e/m}$. Performance differences are more pronounced on either side of this locus: for $(k_e,f)$ pairs above the curve, the RMSE is generally lower than for pairs below it. Immediately below the locus, the RMSE attains a local maximum, showing that the worst performing parameter values lie in close proximity to the best, indicating a sensitivity that may be problematic in practice. 

The ECC maintains satisfactory performance in both nominal and uncertain scenarios. However, a slight performance degradation is observed at low interaction forces (typically occurring at lower frequencies) due to the increased influence of the Stribeck static friction component. To illustrate these results in terms of more conventional metrics, Fig.~\ref{im:soft_ellip} presents the projection of the system trajectory on the $e\times F$ plane for selected values of environmental stiffness and excitation frequency.

\begin{figure}[h!]
    \centering
    \begin{subfigure}{0.49\linewidth}
        \centering
        \includegraphics[width=\linewidth]{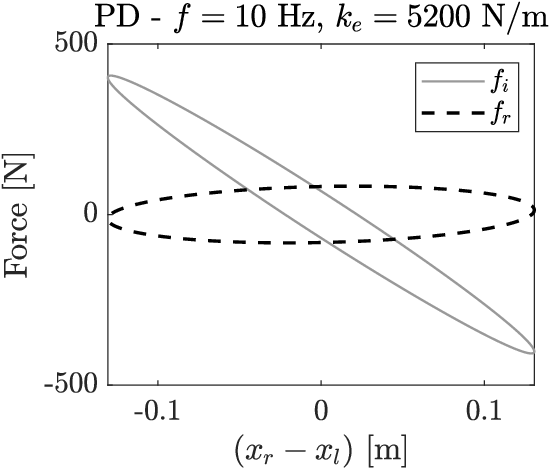}
        \caption{}
    \end{subfigure}
    \hfill
    \begin{subfigure}{0.49\linewidth}
        \centering
        \includegraphics[width=\linewidth]{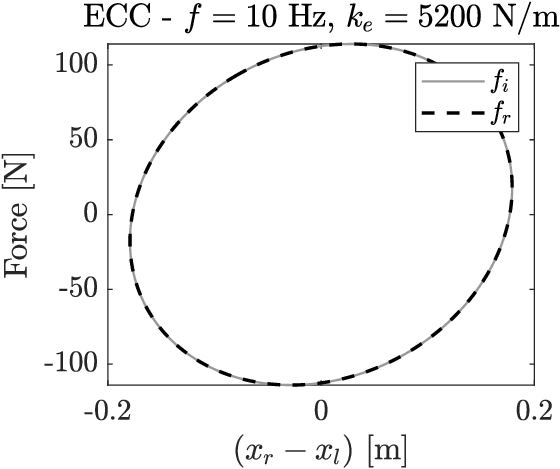}
        \caption{}
    \end{subfigure}
    \medskip
    \begin{subfigure}{0.49\linewidth}
        \centering
        \includegraphics[width=\linewidth]{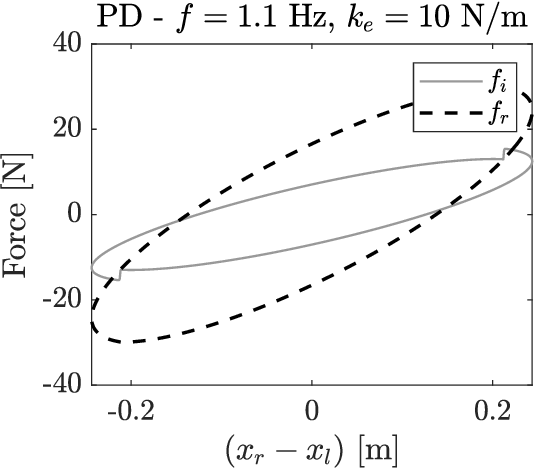}
        \caption{}
    \end{subfigure}
    \hfill
    \begin{subfigure}{0.49\linewidth}
        \centering
        \includegraphics[width=\linewidth]{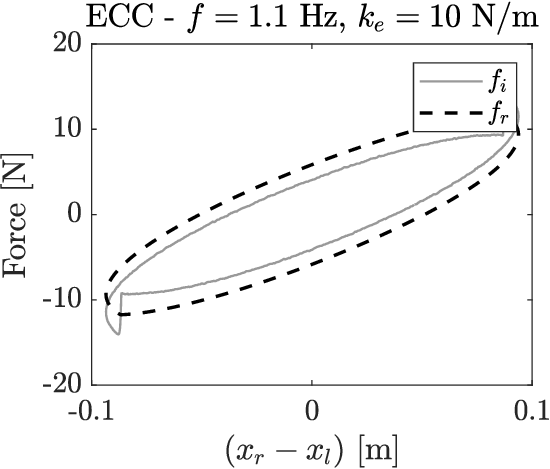}
        \caption{}
    \end{subfigure}
    \caption{Projection of the system trajectory on the $e\times F$ plane: (a), (b) nominal model; (c), (d) Stribeck friction and measurement noise.}
    \label{im:soft_ellip}
\end{figure}

\section{Conclusion}

This paper addressed the problem of achieving a desired impedance in compliant control problems using the admittance structure and linear positional controllers. By establishing a rigorous equivalence with generalized output-regulation theory, we showed that under those assumptions (admittance control and linear position control), there exists a unique control law that attains exact compliant control. Simulations demonstrate that the extra terms, relative to a conventional PD controller,  improve performance in both the ideal case and in the presence of uncertainties. Future work will focus on experimental implementation, extension to multi-degree-of-freedom systems, and the assessment of integrating adaptive mechanisms into the derived control law.

\begin{ack}

\end{ack}

\bibliographystyle{plain}        
\bibliography{references}           



\appendix
\section{Proof of Lemma \ref{lem:equi}}\label{app:lem1}

Subtracting equations \eqref{eq:ref_impedance} and \eqref{eq:aux} leads to
\begin{equation}
    w:=f_r-f_i=m_d\ddot{\hat{e}}+b_d\dot{\hat{e}}+k_d\hat{e}.
\end{equation}
Additionally, Assumption \ref{ass:OSP} implies $m_d\ge0$, $k_d\ge0$ and $b_d>0$. 

\textbf{Case 1: $k_d\neq0$.} If either $m_d=0$ or $m_d\ne0$, the polynomial
$p(s)=m_d s^2 + b_d s + k_d$ is Hurwitz. Therefore, bounded $w$ with $w(t)\to0$ yields $\hat e(t)\to0$. The converse follows from the uniform continuity of $\dot{\hat{e}}$ and $\ddot{\hat{e}}$. Under this hypothesis, Barbalat's Lemma \cite{khalil2002} ensures that $\hat{e}\to0$ implies $\dot{\hat{e}}\to0$ and $\ddot{\hat{e}}\to0$. The convergence of $w(t)$ to zero then follows directly. 

\textbf{Case 2: $k_d=0$.} The relation reduces to
\begin{equation}
w=m_d\ddot{\hat{e}}+b_d\dot{\hat{e}}.
\end{equation}
Let $y:=\dot{\hat{e}}$. Then
\begin{equation}
w=m_d\dot{y} + b_dy.
\end{equation}
If $m_d>0$, the same argument presented in Case 1 shows that bounded $w$ with $w(t)\to0$ implies $\dot{\hat{e}}\to0$. If $m_d=0$, then one obtains the algebraic relation $b_dy = w$. That also implies bounded $w$ with $w(t)\to0$ yields $y=\dot{\hat{e}}\to0$. The converse follows from the same argument of Case 1.

\section{Proof of Proposition \ref{prop:bounded_state}}\label{app:prop1}
Under the control law \eqref{eq:sol}, the dynamics of the actuator are

\begin{equation}\label{eq:cl_act}
    m_a\ddot{\hat{e}}+k_2\dot{\hat{e}}+k_1\hat{e}=0.
\end{equation}

Assuming $x_r(t)=0$, the combination of \eqref{eq:cl_act}, \eqref{eq:load_dyn} and \eqref{eq:soft_gen} with $m_d=0$ leads to the linear dynamics

\begin{equation}\label{eq:cl_lin}
    \dot{x}=A_{cl}x
\end{equation}
with $x$ given by \eqref{eq:scomplete} and $A_{cl}$ given by
\begin{equation}
    A_{cl}= \begin{bmatrix}
        0 & 1 & 0 & 0 & 0 \\
        -k_1/m_a & -k_2/m_a & 0 & 0 & 0 \\
        0 & 0 & 0 & 1 & 0 \\
        k/m_l & b/m_l & -k_d/m_l & -b_d/m_l & 0\\
        0 & 0 & (k-k_d)/b & (b-b_d)/b & -k/b\\
    \end{bmatrix}.
\end{equation}

The eigenvalues of $A_{cl}$ are

\begin{equation}\label{eq:eig}
\begin{aligned}
& \lambda_{1,2}=\frac{-b_d\pm\sqrt{b_d^2-4k_dm_l}}{2m_l} \\
& \lambda_{3,4}=\frac{-k_2\pm\sqrt{k_2^2-4k_1m_a}}{2m_a} \\
& \lambda_{5}=-\frac{k}{b}.
\end{aligned}
\end{equation}

Given that all parameters are positive, \eqref{eq:eig} shows that $A_{cl}$ is Hurwitz. Since $x_r(t),\dot{x}_r(t)\ne0$ would make \eqref{eq:cl_lin} a linear affine system with input $u=[x_r \,\,\dot{x}_r]^\intercal$, the state $x(t)$ is bounded for any bounded and sufficiently smooth signal $x_r(t)$.

\section{Transfer functions from $x_r$ and its derivatives to $\hat{e}$}\label{app:trans} 

For the PD controller, the closed-loop system with interaction force \eqref{eq:rigidenv} can be described by the matrices:

\begin{equation}
\begin{aligned}
 &      A=\begin{bmatrix}
        0 & 1 & 0 & 0\\
        -(k_e+k_1)/m & -k_2/m & k_1/m & k_2/m\\
        0 & 0 & 0 & 1\\
        k_e/m & b/m_l & -k_d/m_l & -b_d/m_l\\
    \end{bmatrix},\\
&     B = \begin{bmatrix}
    0 & 0 & 0 \\
    0 & 0 & 0 \\
    0 & 0 & 0 \\
    k_d/m_d & b_d/m_d & 1
    \end{bmatrix}, \quad C=\left[\begin{array}{c c c c}
    -1 & 0 & 1 & 0 \\
    \end{array}\right],
\end{aligned}
\end{equation}
with the state, input and output given by $x=[x_e\,\,\dot{x}_e\,\,\hat{x}\,\,\dot{\hat{x}}]^\intercal$, $u=[x_r\,\,\dot{x}_r\,\,\ddot{x}_r]$, and $y=\hat{e}$. Then, the input-output behavior of this system is described by $G(s)=C(sI-A)^{-1}B$, with

\begin{equation}
    G(s)=\frac{1}{D(s)}\left[\begin{array}{c c c}
    N_1(s) & N_2(s) & N_3(s)
    \end{array}\right],
\end{equation}
in which 
\begin{equation}
    N_1=\frac{k_d}{b_d}N_2=\frac{k_d}{m_d}N_3=k_d(ms^2+k_e),
\end{equation}
and
\begin{dmath}
    D(s)=mm_ds^4+(b_dm+k_2m_d)s^3+\left(m_d(k_e+k_1)+b_dk_2+k_dm\right)s^2+\left(b_d(k_e+k_1)+k_2(k_e+k_d)\right)s+k_1(k_e+k_d)+k_ek_d.
\end{dmath}
\end{document}